\documentclass[11pt]{article}
\usepackage[T2A]{fontenc}
\usepackage{ alphalph,etoolbox }
\usepackage{amsthm, amsmath, amssymb, amsbsy, amscd, amsfonts, latexsym, euscript,epsf, bm}
\newtheorem{thm}{Theorem}[section]
\usepackage{bbold, bm}
\usepackage{epic,eepic}
\usepackage{texdraw}
\usepackage[dvips]{color}
\usepackage{color}
\usepackage{ dsfont }
\usepackage{graphicx}
\usepackage{exscale,relsize}
\usepackage{authblk}
\usepackage{url}
\usepackage{hyperref}
\usepackage{enumerate}
\usepackage{graphicx}
\newtheorem{proposition}[thm]{Proposition}

\newtheorem{theorem}[thm]{Theorem}

\theoremstyle{definition}

\newtheorem{remark}[thm]{Remark}

\DeclareMathOperator{\id}{id}

\title{Electric network and Hirota type $4$-simplex maps}
\date{}

\author{S. Konstantinou-Rizos\thanks{skonstantin84@gmail.com}}
\affil{Centre of Integrable Systems, P.G. Demidov Yaroslavl State University, Yaroslavl, Russia}

\patchcmd{\subequations}{\alph{equation}}{\alphalph{\value{equation}}}{}{}

\begin{document}

\maketitle

\begin{abstract}
Bazhanov--Stroganov (4-simplex) maps are set-theoretical solutions to the 4-simplex equation, namely the fourth member of the family of $n$-simplex equations, which are fundamental equations of mathematical physics. 
In this paper, we develop a method for constructing Bazhanov--Stroganov maps as extensions of tetrahedron maps which are set-theoretical solutions to the Zamolodchikov tetrahedron (3-simplex) equation. We employ this method to construct birarional Bazhanov--Stroganov maps which boil down to the famous electric network and Hirota tetrahedron maps at a certain limit.
\end{abstract}

\bigskip
\begin{quotation}
\end{quotation}
\begin{quotation}
\end{quotation}
\begin{quotation}
\end{quotation}

\section{Introduction}\label{intro}
The $n$-simplex equations are generalisations of the famous Yang--Baxter equation and are fundamental equations of Mathematical Physics. The most celebrated members of the family   of $n$-simplex equations are the Yang--Baxter (2-simplex) equation, the Zamolodchikov tetrahedron (3-simplex) equation \cite{Zamolodchikov-2} and the Bazhanov--Stroganov (4-simplex) equation \cite{Bazhanov}. 

The popularity of $n$-simplex equations is due to the fact that they appear in a wide range of areas of Mathematics and Physics, and they are strictly related to integrable systems of differential and difference equations. In particular, the $n$-simplex equations have applications in statistical mechanics, quantum field theories, combinatorics, low-dimensional topology,  the theory of integrable systems, as well as a plethora of other  fields of science (see, e.g., \cite{Bazhanov-Sergeev, Bazhanov-Sergeev-2, Sergeev-2, Kapranov-Voevodsky, Kashaev-Sergeev, Kassotakis-Kouloukas, Kassotakis-Tetrahedron, Nijhoff, TalalUMN21}). Therefore, the construction and classification of $n$-simplex maps is quite relevant and constitutes a very active area of research.

There are several methods that associate $n$-simplex maps (set theoretical solutions to the $n$-simplex equation) with integrable systems. We indicatively refer to the relation between $n$-simplex maps and integrable lattice equations via symmetries \cite{Pap-Tongas, Kassotakis-Tetrahedron} as well  as integrable nonlinear PDEs via Darboux and B\"acklund transformations \cite{BIKRP, Sokor-Sasha, Sokor-Papamikos}. Thus, the development of methods for constructing interesting $n$-simplex maps is a significant  task which may give rise to important integrable models.

This paper is concerned with the development of a simple scheme for construction of Bazhanov--Stroganov 4-simplex maps as extensions of Zamolodchikov tetrahedron maps. The proposed scheme is a generalisation of the method presented in \cite{Sokor-2023-PhysD} for constructing 4-simplex maps, and its methodology involves working with simple matrix refactorisation problems and makes use of straightforward computational algebra. The advantage of the proposed generalised scheme versus the method presented in \cite{Sokor-2023-PhysD} is that it derives more interesting 4-simplex maps which are noninvolutive extensions of involutive 3-simplex maps. This is a quite interesting phenomenon, since involutive maps possess trivial dynamics. 

As an illustrative example for our method, we first use the Hirota tetrahedron map  which has various important applications \cite{Doliwa-Kashaev,Sergeev}, however it lacks the property of being noninvolutive. One of its 4-simplex extensions, which we present in this paper, has equally elegant form and preserves all the properties of the original map (Lax representation, first integrals etc.); nonetheless, it has the extra significant property of being noninvolutive. Furthermore, we apply our  method  to Kashaev's electric network tetrahedron map \cite{Kashaev, Kashaev-Sergeev},  and we  construct its  4-simplex extensions. The Kashaev electric network transform is the equivalency condition of two electric devices, each made of three resistors and with three outer contacts, with star and triangle diagrams, respectively \cite{Kashaev}. Moreover, it is related to the well-celebrated Miwa's integrable equation, which at a certain continuous limit reduces to the KP (Kadomtsev--Petviashvili) equation \cite{Miwa, Kashaev}.

The rest of the text is organised as follows: In the next section, we provide all the required definitions and statements for the text to be self-contained. In particular, we give the definitions of Zamolodchikov's tetrahedron maps and Bazhanov--Stroganov 4-simplex maps, and we explain  their relation with the local Yang--Baxter and tetrahedron equation, respectively. Section \ref{extension scheme} deals with the development of a simple scheme for constructing Bazhanov--Stroganov maps as  extensions of Zamolodchikov tetrahedron maps. We apply this scheme to a Sergeev type map and construct new 4-simplex maps. In Section \ref{Hirota_maps}, using the method presented in Section \ref{extension scheme}, we derive novel Bazhanov--Stroganov maps which are 4-simplex extensions of the famous Hirota tetrahedron map \cite{Doliwa-Kashaev, Sergeev}. In Section \ref{EN_maps}, we construct new 4-simplex  maps which  can be restricted to the  famous electric network tetrahedron map  at certain limits. Finally, in Section \ref{conclusions} we discuss the obtained results and list possible directions for future research.

\section{Preliminaries}
In this section, we give the  definitions of tetrahedron and Bazhanov--Stroganov maps, and we explain how to derive such maps using matrix refactorisation problems.

\subsection{Local Yang--Baxter equation and tetrahedron maps}
Let $\mathcal{X}$ be a set. We denote  $\mathcal{X}^n =\underbrace{\mathcal{X}\times \ldots \times \mathcal{X}}_\text{$n$}$. A map $T:\mathcal{X}^3\rightarrow \mathcal{X}^3$, namely $T:(x,y,z)\mapsto (u(x,y,z),v(x,y,z),w(x,y,z)),$ is called a \textit{3-simplex map} or \textit{Zamolodchikov map} or \textit{tetrahedron map} if it satisfies the \textit{functional tetrahedron} or \textit{Zamolodchikov's tetrahedron} equation \cite{Zamolodchikov, Zamolodchikov-2}
\begin{equation}\label{Tetrahedron-eq}
    T^{123}\circ T^{145} \circ T^{246}\circ T^{356}=T^{356}\circ T^{246}\circ T^{145}\circ T^{123}.
\end{equation}
Functions $T^{ijk}:\mathcal{X}^6\rightarrow \mathcal{X}^6$, $i,j=1,\ldots 6,~i< j<k$, in \eqref{Tetrahedron-eq} are maps that act as map $T$ on the $ijk$ terms of the Cartesian product $\mathcal{X}^6$ and trivially on the others. For instance,
$T^{356}(x_1,x_2,x_3,x_4,x_5,x_6)=(x_1,x_2,u(x_3,x_5,x_6),x_4,v(x_3,x_5,x_6),w(x_3,x_5,x_6))$.

Now, let ${\rm L}={\rm L}(x)$ be a matrix depending on a variable $x\in\mathcal{X}$ of the form {\small ${\rm L}(x)= \begin{pmatrix} a(x) & b(x)\\  c(x) & d(x)\end{pmatrix},$} where its entries $a, b, c$ and $d$ are scalar functions of $x$. Let ${\rm L}^3_{ij}$, $i,j=1,2, 3$, $i\neq j$, be the $3\times 3$ extensions of ${\rm L}(x)$,  defined by {\small
${\rm L}^3_{12}(x)=\begin{pmatrix} 
 a(x) &  b(x) & 0\\ 
c(x) &  d(x) & 0\\
0 & 0 & 1
\end{pmatrix},~
 {\rm L}^3_{13}(x)= \begin{pmatrix} 
 a(x) & 0 & b(x)\\ 
0 & 1 & 0\\
c(x) & 0 & d(x)
\end{pmatrix}, ~
 {\rm L}^3_{23}(x)=\begin{pmatrix} 
   1 & 0 & 0 \\
0 & a(x) & b(x)\\ 
0 & c(x) & d(x)
\end{pmatrix}$. The following matrix trifactorisation problem
\begin{equation}\label{Lax-Tetra}
    {\rm L}^3_{12}(u){\rm L}^3_{13}(v){\rm L}^3_{23}(w)= {\rm L}^3_{23}(z){\rm L}^3_{13}(y){\rm L}^3_{12}(x),
\end{equation}
is the Maillet--Nijhoff equation \cite{Nijhoff} in Korepanov's form, which appears in the literature as the \textit{local Yang--Baxter} equation. If a map of $T:\mathcal{X}^3\rightarrow \mathcal{X}^3$ satisfies the local Yang--Baxter equation \eqref{Lax-Tetra}, then this map is possibly a tetrahedron map, and equation \eqref{Lax-Tetra} is called its \textit{Lax representation} \cite{Dimakis-Hoissen}.

The local Yang--Baxter \eqref{Lax-Tetra} is a generator of Zamolodchikov tetrahedron maps. Tetrahedron maps are related to pentagon maps \cite{Doliwa-Kashaev, Kashaev-Sergeev-2}, and, generally, $n$-simplex maps are related to solutions of the $n$-gon equations \cite{Dimakis-Hoissen}. The most popular tetrahedron maps appear in the works of Sergeev \cite{Sergeev} and Kashaev--Korepanov--Sergeev \cite{Kashaev-Sergeev}.

\subsection{Local tetrahedron equation and 4-simplex maps}
A map $S:\mathcal{X}^4\rightarrow \mathcal{X}^4$, namely $S:(x,y,z,t)\mapsto (u(x,y,z,t),v(x,y,z,t),w(x,y,z,t),r(x,y,z,t)),$
is called a \textit{4-simplex map} or\textit{Bazhanov--Stroganov} map if it satisfies the \textit{set-theoretical 4-simplex} equation \cite{Bazhanov}
\begin{equation}\label{4-simplex-eq}
    S^{1234}\circ S^{1567} \circ S^{2589}\circ S^{368,10} \circ S^{479,10}=S^{479,10}\circ S^{368,10}\circ S^{2589}\circ S^{1567} \circ S^{1234}.
\end{equation}
Functions $S^{ijkl}:\mathcal{X}^{10}\rightarrow\mathcal{X}^{10}$, $i,j,k,l=1,\ldots 10,~i< j<k<l$, in \eqref{4-simplex-eq} are maps that act as map $S$ on the $ijkl$ terms of the Cartesian product $\mathcal{X}^{10}$ and trivially on the others. For instance,
\begin{align*}
S^{1567}&(x_1,x_2,\ldots, x_{10})=\\
&(u(x_1,x_5,x_6,x_7),x_2,x_3,x_4,v(x_1,x_5,x_6,x_7),w(x_1,x_5,x_6,x_7),r(x_1,x_5,x_6,x_7),x_8,x_9,x_{10}).
\end{align*}

Now, let ${\rm L}={\rm L}(x)$ be a $3\times 3$ square matrix depending on a variable $x\in\mathcal{X}$ of the form
{\small ${\rm L}(x)= \begin{pmatrix} a(x) & b(x) & c(x)\\ d(x) & e(x) & f(x)\\ k(x) & l(x) & m(x)\end{pmatrix},$}
where its entries are scalar functions of $x$. Let ${\rm L}^6_{ijk}(x)$, $i,j,k=1,\ldots 6$, $i< j<k$, be the $6\times 6$ extensions of matrix ${\rm L}(x)$,  defined by
{\small
\begin{subequations}\label{6x6-extensions}
\begin{align}
    & {\rm L}^6_{123}(x)= \begin{pmatrix} 
a(x) & b(x) & c(x) & 0 & 0 & 0\\ 
d(x) & e(x) & f(x) & 0 & 0 & 0\\
k(x) & l(x) & m(x) & 0 & 0 & 0\\
0 & 0 & 0 & 1 & 0 & 0 \\
0 & 0 & 0 & 0 & 1 & 0 \\
0 & 0 & 0 & 0 & 0 & 1 \\
\end{pmatrix}, \quad
{\rm L}^6_{145}(x)= \begin{pmatrix} 
a(x) & 0 & 0 & b(x) & c(x) & 0\\ 
0 & 1 & 0 & 0 & 0 & 0 \\
0 & 0 & 1 & 0 & 0 & 0 \\
d(x) & 0 & 0 & e(x) & f(x) & 0\\
k(x) & 0 & 0 & l(x) & m(x) & 0\\
0 & 0 & 0 & 0 & 0 & 1 \\
\end{pmatrix},\\
& {\rm L}^6_{246}(x)= \begin{pmatrix} 
1 & 0 & 0 & 0 & 0 & 0 \\
0 & a(x) & 0 & b(x) & 0 & c(x)\\ 
0 & 0 & 1 & 0 & 0 & 0 \\
0 & d(x) & 0 & e(x) & 0 & f(x)\\
0 & 0 & 0 & 0 & 1 & 0 \\
0 & k(x) & 0 & l(x) & 0 & m(x)
\end{pmatrix},\quad 
 {\rm L}^6_{356}(x)= \begin{pmatrix} 
 1 & 0 & 0 & 0 & 0 & 0 \\
 0 & 1 & 0 & 0 & 0 & 0 \\
0 & 0 & a(x) & 0 & b(x) & c(x)\\ 
0 & 0 & 0 & 1 & 0 & 0 \\
0 & 0 & d(x) & 0 & e(x) & f(x)\\
0 & 0 & k(x) & 0 & l(x) & m(x)
\end{pmatrix}.
\end{align}
\end{subequations}
}

We call the following matrix four-factorisation problem
\begin{equation}\label{local-tetrahedron}
    {\rm L}^6_{123}(u){\rm L}^6_{145}(v){\rm L}^6_{246}(w){\rm L}^6_{356}(r)={\rm L}^6_{356}(t){\rm L}^6_{246}(z){\rm L}^6_{145}(y){\rm L}^6_{123}(x)
\end{equation}
\textit{local tetrahedron equation}. The local tetrahedron equation is a generator of potential solutions to the 4-simplex equation. If map $S:\mathcal{X}^4\rightarrow\mathcal{X}^4$ satisfies equation \eqref{local-tetrahedron}, then this  map is possibly a 4-simplex map, and the matrix refactorisation problem \eqref{local-tetrahedron} is called a \textit{Lax representation} for map $S$ \cite{Dimakis-Hoissen}.

\section{4-simplex extension scheme}\label{extension scheme}
Let $T:\mathcal{X}^3\rightarrow \mathcal{X}^3$ be a tetrahedron map with Lax representation \eqref{Lax-Tetra} for some matrix ${\rm L}(x)$. The aim is to construct systematically a 4-simplex map $S:\mathcal{X}^4\rightarrow \mathcal{X}^4$, with Lax representation \eqref{local-tetrahedron}, such that map $S$ implies $T$ at a certain limit. This can be achieved using a simple scheme which is summarised in Figure \ref{4-simplex scheme}.

In particular:

\textbf{Step I:} Consider a tetrahedron map, $T:\mathcal{X}^3\rightarrow\mathcal{X}^3$, generated by ${\rm L}(x)$ via \eqref{Lax-Tetra}. If we consider the $3\times 3$ extension of  ${\rm L}(x)$, namely matrix {\small ${\rm M}(x)\equiv{\rm L}^3_{12}(x)= \begin{pmatrix} 
a(x) & b(x) & 0\\ 
c(x) & d(x) & 0 \\
0 & 0 & 1
\end{pmatrix}$}, and substitute it to the local tetrahedron equation \eqref{local-tetrahedron}, we will obtain a trivial extension of map $T$ as a solution to the 4-simplex equation \cite{Sokor-2023-PhysD}.

\textbf{Step II:} In order to construct an nontrivial extension of map $T$, we introduce an auxiliary variable $x_2$, namely we consider the following matrix
\begin{equation}\label{K-matrix}
{\rm K}(x_1,x_2)= \begin{pmatrix} 
a(x_1,x_2) & b(x_1,x_2) & 0\\ 
c(x_1,x_2) & d(x_1,x_2) & 0 \\
0 & 0 & x_2
\end{pmatrix},
\end{equation}
such that for $x_2\rightarrow 1$, we have ${\rm K}(x_1,x_2)\rightarrow {\rm M}(x_1)$. Substitute ${\rm K}(x_1,x_2)$ to the local tetrahedron equation
\begin{equation}\label{local-tetra}
  {\rm K}^6_{123}(u_1,u_2){\rm K}^6_{145}(v_1,v_2){\rm K}^6_{246}(w_1,w_2){\rm K}^6_{356}(r_1,r_2)={\rm K}^6_{356}(t_1,t_2;\delta){\rm K}^6_{246}(z_1,z_2){\rm K}^6_{145}(y_1,y_2){\rm K}^6_{123}(x_1,x_2).
\end{equation}

\textbf{Step III:} Solve \eqref{local-tetra} for $u_i$, $v_i$, $w_i$ and $r_i$, $i=1,2$; we  aim to obtain a 4-simplex map or a correspondence which, for particular values of the free variables, will define 4-simplex maps. 

\begin{figure}[ht]
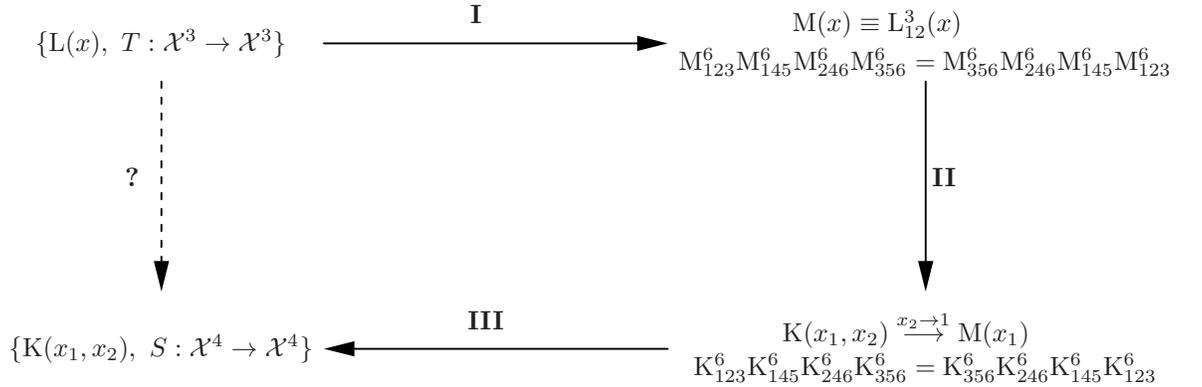

\centering
\centertexdraw{ 
\setunitscale 0.5
\move (-3.3 2.2)  \arrowheadtype t:F \avec(0.3 2.2)
\move (-5 1.8) \lpatt(0.067 0.1) \arrowheadtype t:F \avec(-5 -.4) \lpatt()
\move (3 1.8)  \arrowheadtype t:F \avec(3 -.4) 
\move (.3 -1)  \arrowheadtype t:F \avec(-3.3 -1) 
\textref h:C v:C \htext(-5 2.2){\small $\{{\rm L}(x),~T:\mathcal{X}^3\rightarrow\mathcal{X}^3\}$}
\textref h:C v:C \htext(2.5 2.4){\small ${\rm M}(x)\equiv{\rm L}^3_{12}(x)$}
\textref h:C v:C \htext(3 2){\small ${\rm M}^6_{123}{\rm M}^6_{145}{\rm M}^6_{246}{\rm M}^6_{356}=
  {\rm M}^6_{356}{\rm M}^6_{246}{\rm M}^6_{145}{\rm M}^6_{123}$}
\textref h:C v:C \htext(2.8 -.8){\small ${\rm K}(x_1,x_2)\stackrel{x_2\rightarrow 1}{\longrightarrow}{\rm M}(x_1)$}
\textref h:C v:C \htext(3 -1.2){\small ${\rm K}^6_{123}{\rm K}^6_{145}{\rm K}^6_{246}{\rm K}^6_{356}=
  {\rm K}^6_{356}{\rm K}^6_{246}{\rm K}^6_{145}{\rm K}^6_{123}$}
\textref h:C v:C \htext(-5 -1){\small $\{{\rm K}(x_1,x_2),~S:\mathcal{X}^4\rightarrow\mathcal{X}^4\}$}
\textref h:C v:C \small{\htext(-5.3 .8){\textbf{?}}}
\textref h:C v:C \small{\htext(3.2 .8){\textbf{II}}}
\textref h:C v:C \small{\htext(-1.7 2.5){\textbf{I}}}
\textref h:C v:C \small{\htext(-1.6 -.7){\textbf{III}}}
}
\caption{4-simplex extension scheme.}\label{4-simplex scheme}
\end{figure}

\begin{remark}\normalfont
This extends the method presented in \cite{Sokor-2023-PhysD} by allowing the elements $a, b, c$  and $d$ of matrix ${\rm K}(x_1,x_2)$ in \eqref{K-matrix} to depend on the auxiliary variable $x_2$. We will demonstrate that the latter assumption implies more interesting 4-simplex maps with more interesting dynamics than the maps derived via the  method in \cite{Sokor-2023-PhysD}.
\end{remark}

\subsection{Example: A new Sergeev type 4-simplex map}
Here,  we demonstrate that, employing the extension scheme of Section \ref{extension scheme}, one may construct more interesting 4-simplex extensions of Sergeev's  maps than the ones presented in \cite{Sokor-2023-PhysD}. 

Specifically: 

\textbf{Step I:} Consider the $k$-parametric family of tetrahedron maps $\mathbb{C}^3\rightarrow \mathbb{C}^3$ \cite{Sergeev}:
\begin{equation}\label{Sergeev-b}
    T:(x,y,z)\rightarrow \left(\frac{xy}{y+xz},\frac{xz}{k},\frac{y+xz}{x}\right),
\end{equation}
with Lax representation \eqref{Lax-Tetra}, where ${\rm L}^3_{ij}$, $i,j=1,2,3$, $i<j$, are the $3\times 3$ generalisations of matrix
${\rm L}(x)=\begin{pmatrix} 1 & x\\ \frac{\kappa}{x} & 0 \end{pmatrix}$, $x\in\mathbb{C}$. If we consider the natural $3\times 3$ extension of matrix ${\rm L}(x;\kappa)$, namely matrix  ${\rm M}(x)\equiv {\rm L}^3_{12}(x;\kappa)=\begin{pmatrix} 1 & x & 0\\  \frac{\kappa}{x} & 0 & 0\\ 0 & 0 & 1 \end{pmatrix}$, and substitute it to the local tetrahedron equation \eqref{local-tetra}, one  will  obtain the 4-simplex map 
$S:(x,y,z,t)\rightarrow \left(\frac{xy}{y+xz},\frac{xz}{k},\frac{y+xz}{x},t\right),$
which  is a trivial extension of \eqref{Sergeev-b}.

\textbf{Step II:} Consider, instead, matrix {\small ${\rm K}(x_1,x_2)=\begin{pmatrix} 
1 & \frac{x_1}{x_2} & 0\\ 
\frac{\kappa}{x_1} & 0 & 0 \\
0 & 0 & x_2
\end{pmatrix}$}, $x_i\in\mathbb{C}$, $i=1,2$, $\det({\rm K}(x_1,x_2))=-k$ and substitute it to the local tetrahedron equation
\begin{equation}\label{Sergeev-LT}
    {\rm K}^6_{123}(u_1,u_2){\rm K}^6_{145}(v_1,v_2){\rm K}^6_{246}(w_1,w_2){\rm K}^6_{356}(r_1,r_2)={\rm K}^6_{356}(t_1,t_2){\rm K}^6_{246}(z_1,z_2){\rm K}^6_{145}(y_1,y_2){\rm K}^6_{123}(x_1,x_2).
\end{equation}

\textbf{Step III:} We solve equation  \eqref{Sergeev-LT} for $(u_1,u_2,v_1,v_2,w_1,w_2,r_1,r_2)$, and we see  that equation \eqref{Sergeev-LT} is equivalent to the following correspondence between $\mathbb{C}^8$ and $\mathbb{C}^8$
\begin{align*}
u_1=\frac{x_1y_1z_2}{x_1z_1+y_1z_2},\quad u_2=x_2,\quad v_1=\frac{x_1z_1}{k},\quad v_2=\frac{y_2z_2}{w_2},\quad w_1=\frac{w_1x_2(x_1z_1+y_1z_2)}{x_1y_2z_2},\quad r_1=\frac{t_1y_2z_2}{w_2x_2},\quad r_2=\frac{t_2z_2}{w_2}.
\end{align*}
The above system does not define a 4-simplex map $\mathbb{C}^8\rightarrow \mathbb{C}^8$ for arbitrary $w_2$. However, for the choices $w_2=t_2$ and $w_2=y_2$ we obtain the following Bazhanov--Stroganov 4-simplex maps
\begin{align}
    &S_1: (x_1,x_2,y_1,y_2,z_1,z_2,t_1,t_2)\rightarrow \left(\frac{x_1y_1z_2}{x_1z_1+y_1z_2}, x_2, \frac{x_1z_1}{k},\frac{y_2z_2}{t_2}, \frac{x_2t_2(x_1z_1+y_1z_2)}{x_1y_2z_2}, t_2, \frac{t_1y_2z_2}{x_2t_2}, z_2\right),\label{Sergeev-map-1}\\
    &S_2: (x_1,x_2,y_1,y_2,z_1,z_2,t_1,t_2)\rightarrow \left(\frac{x_1y_1z_2}{x_1z_1+y_1z_2}, x_2, \frac{x_1z_1}{k},z_2, \frac{x_2(x_1z_1+y_1z_2)}{x_1z_2}, y_2, \frac{t_1z_2}{x_2}, \frac{t_2z_2}{y_2}\right).\label{Sergeev-map-2}
\end{align}

Maps \eqref{Sergeev-map-1} and \eqref{Sergeev-map-2} are more interesting than maps (27) and (28) in \cite{Sokor-2023-PhysD}.  This can  be done for all maps derived in  \cite{Sokor-2023-PhysD}. It can be proven that maps \eqref{Sergeev-map-1} and \eqref{Sergeev-map-2} are birational, and their inverses are also 4-simplex maps.

\begin{remark}\normalfont
Maps (27) and (28) in \cite{Sokor-2023-PhysD} were generated by {\small ${\rm K}(x_1,x_2)=\begin{pmatrix} 
1 & x_1 & 0\\ 
\frac{\kappa}{x_1} & 0 & 0 \\
0 & 0 & x_2
\end{pmatrix}$}.
\end{remark}

\section{Hirota type Bazhanov--Stroganov map}\label{Hirota_maps}
In this section, we apply the 4-simplex extension scheme of the previous section to the  well-known Hirota tetrahedron map. We construct novel 4-simplex maps which can be restricted to the Hirota map at certain limit.

\subsection{Hirota 4-simplex map}
The famous Hirota tetrahedron map  reads \cite{Doliwa-Kashaev, Sergeev}
\begin{equation}\label{Hirota}
    T:(x,y,z)\rightarrow \left(\frac{xy}{x+z},x+z,\frac{yz}{x+z}\right),
\end{equation}
and it possesses a Lax representation \eqref{Lax-Tetra} for ${\rm L}(x)=\begin{pmatrix} 
x & 1\\ 
1 & 0
\end{pmatrix}$, $x\in\mathbb{C}$. 

In order to construct a 4-simplex extension of Hirota map \eqref{Hirota}, we consider the $3\times 3$ extension of ${\rm L}(x)$: {\small 
\begin{equation}\label{L123-Hirota}
{\rm M}(x)\equiv {\rm L}^3_{12}(x)=\begin{pmatrix} 
x & 1 & 0\\ 
1 & 0 & 0\\
0 & 0 & 1
\end{pmatrix}.\end{equation}} Then, we introduce a generalisation of matrix ${\rm M}(x)$, namely matrix {\small ${\rm K}(x_1,x_2)=\begin{pmatrix} \frac{x_1}{x_2} & 1 & 0\\ 1 & 0 & 0 \\ 0 & 0 & x_2\end{pmatrix},$}
such that for $x_2\rightarrow 1$, ${\rm K}(x_1,x_2)\rightarrow {\rm M}(x_1)$. We consider the $6\times 6$ extensions of matrix ${\rm K}(x_1,x_2)$, namely the following
{\small\begin{align*}
    & {\rm K}^6_{123}(x_1,x_2)= \begin{pmatrix} 
\frac{x_1}{x_2} & 1 & 0 & 0 & 0 & 0\\ 
1 & 0 & 0 & 0 & 0 & 0\\
0 & 0 & x_2 & 0 & 0 & 0\\
0 & 0 & 0 & 1 & 0 & 0 \\
0 & 0 & 0 & 0 & 1 & 0 \\
0 & 0 & 0 & 0 & 0 & 1 \\
\end{pmatrix}, \quad
{\rm K}^6_{145}(x_1,x_2)= \begin{pmatrix} 
\frac{x_1}{x_2} & 0 & 0 & 1 & 0 & 0\\ 
0 & 1 & 0 & 0 & 0 & 0 \\
0 & 0 & 1 & 0 & 0 & 0 \\
1 & 0 & 0 & 0 & 0 & 0\\
0 & 0 & 0 & 0 & x_2 & 0\\
0 & 0 & 0 & 0 & 0 & 1 \\
\end{pmatrix},\\
& {\rm K}^6_{246}(x_1,x_2)= \begin{pmatrix} 
1 & 0 & 0 & 0 & 0 & 0 \\
0 & \frac{x_1}{x_2} & 0 & 1 & 0 & 0\\ 
0 & 0 & 1 & 0 & 0 & 0 \\
0 & 1 & 0 & 0 & 0 & 0\\
0 & 0 & 0 & 0 & 1 & 0 \\
0 & 0 & 0 & 0 & 0 & x_2
\end{pmatrix},\quad 
 {\rm K}^6_{356}(x_1,x_2)= \begin{pmatrix} 
 1 & 0 & 0 & 0 & 0 & 0 \\
 0 & 1 & 0 & 0 & 0 & 0 \\
0 & 0 & \frac{x_1}{x_2} & 0 & 1 & 0\\ 
0 & 0 & 0 & 1 & 0 & 0 \\
0 & 0 & 1 & 0 & 0 & 0\\
0 & 0 & 0 & 0 & 0 & x_2
\end{pmatrix},
\end{align*}}
and substitute them to the local tetrahedron equation
$$
{\rm K}^6_{123}(u_1,u_2){\rm K}^6_{145}(v_1,v_2){\rm K}^6_{246}(w_1,w_2){\rm K}^6_{356}(r_1,r_2)={\rm K}^6_{356}(t_1,t_2){\rm K}^6_{246}(z_1,z_2){\rm K}^6_{145}(y_1,y_2){\rm K}^6_{123}(x_1,x_2).
$$
The above is equivalent to the following correspondence between $\mathbb{C}^8$ and $\mathbb{C}^8$:
\begin{equation}\label{corr-hir-inv}
  u_1=\frac{x_1y_1z_2}{x_2z_1+x_1z_2},~~ u_2=y_2,~~ v_1=\frac{x_2z_1+x_1z_2}{z_2},~~ v_2=x_2,~~ w_1=\frac{w_2x_2y_1z_1}{y_2(x_2z_1+x_1z_2)},~~ r_1=\frac{t_1x_2z_2}{w_2y_2},~~ r_2=\frac{t_2z_2}{w_2}.
\end{equation}

For the choices of the free variable $w_2=t_2$ and $w_2=z_2$ the above correspondence defines 4-simplex maps. In particular, we have the following.

\begin{proposition}\label{Case-a}
The following maps
\begin{equation}\label{Hirota-a}
   S_1: (x_1,x_2,y_1,y_2,z_1,z_2,t_1,t_2)\rightarrow \left(\frac{x_1y_1z_2}{x_2z_1+x_1z_2},y_2,\frac{x_2z_1+x_1z_2}{z_2},x_2,\frac{t_2x_2y_1z_1}{y_2(x_2z_1+x_1z_2)},t_2,\frac{t_1x_2z_2}{t_2y_2},z_2\right),
\end{equation}
\begin{equation}\label{Hirota-b}
   S_2: (x_1,x_2,y_1,y_2,z_1,z_2,t_1,t_2)\rightarrow \left(\frac{x_1y_1z_2}{x_2z_1+x_1z_2},y_2,\frac{x_2z_1+x_1z_2}{z_2},x_2,\frac{x_2y_1z_1z_2}{y_2(x_2z_1+x_1z_2)},z_2,\frac{t_1x_2}{y_2},t_2\right),
\end{equation}
are eight-dimensional 4-simplex maps which share the same functionally independent invariants $I_1=x_1y_1$, $I_2=x_2y_2$, $I_3=x_2+y_2$, $I_4=z_2t_2$ and $I_5=z_2+t_2$. Moreover, map $S_1$ admits the invariant $I_6=x_2z_2t_1$, whereas $S_2$ admits the invariant $I_7=x_2t_1$.
\end{proposition}
\begin{proof}
Maps \eqref{Hirota-a} and \eqref{Hirota-b} follow after substitution of $w_2=t_2$ and $w_2=z_2$, respectively, to \eqref{corr-hir-inv}. The 4-simplex property can be verified by straightforward substitution to the 4-simplex equation \eqref{4-simplex-eq}. The invariants are obvious.
\end{proof}

\begin{remark}\normalfont
    At the limit $(x_2,y_2,z_2,t_2)\rightarrow (1,1,1,1)$ maps \eqref{Hirota-a} and \eqref{Hirota-b} tend to
    \begin{equation}\label{Hirota-limit}
    S:(x,y,z,t)\rightarrow \left(\frac{xy}{x+z},x+z,\frac{yz}{x+z},t\right),
\end{equation}
which is a trivial 4-simplex extension of the Hirota map \eqref{Hirota}. Thus, they can be regarded as 4-simplex extensions of the Hirota map \eqref{Hirota}.
\end{remark}

\subsection{Noninvolutive Hirota type 4-simplex map}
Maps $S_1$ and $S_2$ of Proposition \ref{Case-a} are involutive, i.e. $S_i\circ S_i=\id$, $i=1,2$. In order to derive a noninvolutive 4-simplex generalisation of the Hirota map \eqref{Hirota}, we consider another generalisation of the generator \eqref{L123-Hirota}, namely matrix  {\small ${\rm K}(x_1,x_2)=\begin{pmatrix} x_1 & 1 & 0\\ \frac{1}{x_2} & 0 & 0 \\ 0 & 0 & x_2\end{pmatrix},$} such that $x_2\rightarrow 1$, ${\rm K}(x_1,x_2)\rightarrow {\rm M}(x_1)$, and  also $\det({\rm K}(x_1,x_2))=\det({\rm M}(x_1))=-1$.

We substitute matrix ${\rm K}(x_1,x_2)$ to \eqref{local-tetra}, and the latter is equivalent to
\begin{align}
    &u_1v_1=x_1y_1,\quad w_1+\frac{u_1}{w_2}=y_1,\quad \frac{v_1}{u_2}=\frac{x_1}{y_2}+\frac{z_1}{x_2},\quad \frac{1}{u_2w_2}=\frac{1}{y_2},\nonumber\\
    & r_1u_2=t_1x_2,\quad u_2=y_2,\quad \frac{1}{v_2}=\frac{1}{x_2z_2},\quad \frac{v_2}{r_2}=\frac{x_2}{t_2}. \label{sys-Hirota}
\end{align}

\begin{theorem}[Bazhanov--Stroganov Hirota type map]
The following map
\begin{equation}\label{Hirota-noninv}
     \hat{S}: (x_1,x_2,y_1,y_2,z_1,z_2,t_1,t_2)\rightarrow \left(\frac{x_1x_2y_1}{x_1x_2+y_2z_1},y_2,\frac{x_1x_2+y_2z_1}{x_2},x_2z_2,\frac{y_1y_2z_1}{x_1x_2+y_2z_1},1,\frac{t_1x_2}{y_2},t_2z_2\right),
\end{equation}
is an eight-dimensional noninvolutive 4-simplex map with functionally independent invariants $I_1=x_1y_1,$ $I_2=y_2$, $I_3=x_2t_1$ and $I_4=z_2t_2$. Moreover, map \eqref{Hirota-noninv} admits Lax representation \eqref{local-tetra}, where {\small ${\rm K}(x_1,x_2)=\begin{pmatrix} x_1 & 1 & 0\\ \frac{1}{x_2} & 0 & 0 \\ 0 & 0 & x_2\end{pmatrix}$}. Finally, at the limit $(x_2,y_2,z_2,t_2)\rightarrow (1,1,1,1)$, map \eqref{Hirota-noninv} can be  restricted to the Hirota map \eqref{Hirota}.
\end{theorem}
\begin{proof}
System \eqref{sys-Hirota} has a unique solution represented by map \eqref{Hirota-noninv}; thus, map \eqref{Hirota-noninv} admits Lax representation \eqref{local-tetra}. The 4-simplex property can be readily verified by straightforward substitution to the 4-simplex equation \eqref{4-simplex-eq}. Moreover, we have, for instance, $u_1\circ\hat{S}(x_1,x_2,y_1,y_2,z_1,z_2,t_1,t_2)=\frac{x_1(x_1x_2+y_2z_1)}{x_2(x_1+z_1z_2)}\neq x_1$, i.e. $\hat{S}^2\neq\id$. Therefore, map \eqref{Hirota-noninv} is noninvolutive. Finally, regarding the invariants, we have that, in view of \eqref{Hirota-noninv}, $u_1v_1=x_1y_1,$ $v_2=y_2$, $u_2r_1=x_2t_1$ and $w_2 r_2=z_2t_2$. The matrix with  rows $(\nabla I_i)^t$, $i=1,2,3,4,$ has  rank  4, thus  invariants $I_i$, $i=1,2,3,4,$ are functionally independent. Finally, at the limit $(x_2,y_2,z_2,t_2)\rightarrow (1,1,1,1)$, map \eqref{Hirota-noninv} can be  restricted to the Hirota map \eqref{Hirota-noninv}.
\end{proof} 

\section{Electric network type Bazhanov--Stroganov map}\label{EN_maps}
Here, we apply the 4-simplex extension scheme presented in section \ref{extension scheme} to the  famous electric network tetrahedron map \cite{Kashaev, Kashaev-Sergeev}. We construct novel 4-simplex maps which can be restricted to the electric network map at certain limit.

The electric network tetrahedron map reads \cite{Doliwa-Kashaev, Sergeev}
\begin{equation}\label{EN}
    T:(x,y,z)\rightarrow \left(\frac{xy}{x+z+x y z},x+z+xyz,\frac{yz}{x+z+xyz}\right),
\end{equation}
and it possesses a Lax representation \eqref{Lax-Tetra} for ${\rm L}(x)=\begin{pmatrix} 
x & 1+i x\\ 
1-i x & 0
\end{pmatrix}$,  $i=\sqrt{-1}$, $x\in\mathbb{C}$. 

In order to construct a 4-simplex extension of Hirota map \eqref{Hirota}, we consider the $3\times 3$ extension of ${\rm L}(x)$: {\small 
\begin{equation}\label{L123-EN}
{\rm M}(x)\equiv {\rm L}^3_{12}(x)=\begin{pmatrix} 
x & 1+ix & 0\\ 
1-ix & 0 & 0\\
0 & 0 & 1
\end{pmatrix}.\end{equation}} Then, we introduce a generalisation of matrix ${\rm M}(x)$, namely matrix {\small ${\rm K}(x_1,x_2)=\begin{pmatrix} 
x_1x_2 & 1+ix_1x_2 & 0\\ 
1-ix_1x_2 & x_1x_2 & 0\\
0 & 0 & 1
\end{pmatrix},$}
such that for $x_2\rightarrow 1$, ${\rm K}(x_1,x_2)\rightarrow {\rm M}(x_1)$. 

Let us consider the $6\times 6$ extensions of matrix ${\rm K}(x_1,x_2)$, namely the following
{\small\begin{align*}
    & {\rm K}^6_{123}(x_1,x_2)= \begin{pmatrix} 
  x_1x_2 &  1+ix_1x_2 & 0 & 0 & 0 & 0\\ 
 1-ix_1x_2 &  x_1x_2 & 0 & 0 & 0 & 0\\
0 & 0 & 1 & 0 & 0 & 0\\
0 & 0 & 0 & 1 & 0 & 0 \\
0 & 0 & 0 & 0 & 1 & 0 \\
0 & 0 & 0 & 0 & 0 & 1 \\
\end{pmatrix}, \quad
{\rm K}^6_{145}(x_1,x_2)= \begin{pmatrix} 
 x_1x_2 & 0 & 0 &  1+ix_1x_2 & 0 & 0\\ 
0 & 1 & 0 & 0 & 0 & 0 \\
0 & 0 & 1 & 0 & 0 & 0 \\
 1-ix_1x_2 & 0 & 0 &  x_1x_2 & 0 & 0\\
0 & 0 & 0 & 0 & 1 & 0\\
0 & 0 & 0 & 0 & 0 & 1 \\
\end{pmatrix},\\
& {\rm K}^6_{246}(x_1,x_2)= \begin{pmatrix} 
1 & 0 & 0 & 0 & 0 & 0 \\
0 &  x_1x_2 & 0 &  1+ix_1x_2 & 0 & 0\\ 
0 & 0 & 1 & 0 & 0 & 0 \\
0 &  1-ix_1x_2 & 0 &  x_1x_2 & 0 & 0\\
0 & 0 & 0 & 0 & 1 & 0 \\
0 & 0 & 0 & 0 & 0 & 1
\end{pmatrix},\quad 
 {\rm K}^6_{356}(x_1,x_2)= \begin{pmatrix} 
 1 & 0 & 0 & 0 & 0 & 0 \\
 0 & 1 & 0 & 0 & 0 & 0 \\
0 & 0 &  x_1x_2 & 0 &  1+ix_1x_2 & 0\\ 
0 & 0 & 0 & 1 & 0 & 0 \\
0 & 0 &  1-ix_1x_2 & 0 &  x_1x_2 & 0\\
0 & 0 & 0 & 0 & 0 & 1
\end{pmatrix},
\end{align*}}
and substitute them to the local tetrahedron equation
$$
{\rm K}^6_{123}(u_1,u_2){\rm K}^6_{145}(v_1,v_2){\rm K}^6_{246}(w_1,w_2){\rm K}^6_{356}(r_1,r_2)={\rm K}^6_{356}(t_1,t_2){\rm K}^6_{246}(z_1,z_2){\rm K}^6_{145}(y_1,y_2){\rm K}^6_{123}(x_1,x_2).
$$
The above is equivalent to the following correspondence between $\mathbb{C}^8$ and $\mathbb{C}^8$:
\begin{subequations}\label{corr-EN}
\begin{align}
  &u_2=\frac{x_1x_2y_1y_2}{u_1(x_1x_2+z_1z_2+x_1x_2y_1y_2z_1z_2)},\quad v_2=\frac{x_1x_2+z_1z_2+x_1x_2y_1y_2z_1z_2}{v_1},\label{corr-EN-inv-a}\\ 
  &w_2=\frac{y_1y_2z_1z_2}{w_1(x_1x_2+z_1z_2+x_1x_2y_1y_2z_1z_2)},\quad r_2=\frac{t_1t_2}{r_1}.\label{corr-EN-inv-b}
\end{align}
\end{subequations}

For particular choices of the free variables the above correspondence defines 4-simplex maps. Specifically, we have the following.

\begin{theorem}[Bazhanov--Stroganov electric network transform]
The following maps $S_i:\mathbb{C}^8\rightarrow\mathbb{C}^8$, $i=1,2$, given by
\begin{subequations}\label{EN-map-1}
\begin{align}
    x_1&\mapsto u_1=y_1,\label{EN-map-1-a}\\
    x_2&\mapsto u_2=\frac{x_1x_2y_2}{x_1x_2+z_1z_2+x_1x_2y_1y_2z_1z_2},\label{EN-map-1-b}\\
    y_1&\mapsto v_1=x_1,\label{EN-map-1-c}\\
    y_2&\mapsto v_2=\frac{x_1x_2+z_1z_2+x_1x_2y_1y_2z_1z_2}{x_1},\label{EN-map-1-d}\\
    z_1&\mapsto w_1=t_1,\label{EN-map-1-e}\\
    z_2&\mapsto w_2=\frac{y_1y_2z_1z_2}{t_1(x_1x_2+z_1z_2+x_1x_2y_1y_2z_1z_2)},\label{EN-map-1-f}\\
    t_1&\mapsto r_1=z_1,\label{EN-map-1-g}\\
    t_2&\mapsto r_2=\frac{t_1t_2}{z_1},\label{EN-map-1-h}
\end{align}
\end{subequations}
and
\begin{subequations}\label{EN-map-2}
\begin{align}
    x_1&\mapsto u_1=x_1y_1,\label{EN-map-2-a}\\
    x_2&\mapsto u_2=\frac{x_2y_2}{x_1x_2+z_1z_2+x_1x_2y_1y_2z_1z_2},\label{EN-map-2-b}\\
    y_1&\mapsto v_1=1,\label{EN-map-2-c}\\
    y_2&\mapsto v_2=x_1x_2+z_1z_2+x_1x_2y_1y_2z_1z_2,\label{EN-map-2-d}\\
    z_1&\mapsto w_1=t_1,\label{EN-map-2-e}\\
    z_2&\mapsto w_2=\frac{y_1y_2z_1z_2}{t_1(x_1x_2+z_1z_2+x_1x_2y_1y_2z_1z_2)},\label{EN-map-2-f}\\
    t_1&\mapsto r_1=z_1,\label{EN-map-2-g}\\
    t_2&\mapsto r_2=\frac{t_1t_2}{z_1},\label{EN-map-2-h}
\end{align}
\end{subequations}
respectively, are eight-dimensional 4-simplex maps. Map $S_1$ has the  following functionally independent invariants:
$$
I_1=x_1y_1,\quad I_2=x_1+y_1,\quad I_3=x_2y_2,\quad I_4=z_1t_1,
$$
whereas map $S_2$ admits the functionally independent invariants:
$$
J_1=x_1y_1,\quad J_2=x_2y_2,\quad J_3=z_1t_1,\quad I_4=z_1+t_1.
$$

Maps $S_i$, $i=1,2$, share the Lax representation \eqref{local-tetra}, where {\small ${\rm K}(x_1,x_2)=\begin{pmatrix} x_1x_2 & 1+ix_1x_2 & 0\\ 1-ix_1x_2 & x_1x_2 & 0 \\ 0 & 0 & x_2\end{pmatrix}$}. Furthermore, map $S_2$ is noninvolutive.

Finally, at the limit $(x_1,y_1,z_1,t_1)\rightarrow(1,1,1,1)$, maps $S_i$, $i=1,2$, can  be restricted  to the electric network map Zamolodchikov  tetrahedron map  \eqref{EN}.
\end{theorem}

\begin{proof}
    Maps \eqref{EN-map-1} and \eqref{EN-map-2} are obtained for the following choices of the free parameters $u_1$, $v_1$, $w_1$ and $r_1$ of correspondence \eqref{corr-EN}:
    $$
    u_1=y_1,\quad v_1=x_1,\quad w_1=t_1,\quad r_1=z_1,
    $$
    and
    $$
    u_1=x_1y_1,\quad v_1=1,\quad w_1=t_1,\quad r_1=z_1,
    $$
    respectively. Moreover, by construction, they satisfy the matrix  refactrorisation problem  \eqref{local-tetra}, for {\small ${\rm K}(x_1,x_2)=\begin{pmatrix} x_1x_2 & 1+ix_1x_2 & 0\\ 1-ix_1x_2 & x_1x_2 & 0 \\ 0 & 0 & x_2\end{pmatrix}$}. Moreover, it can  be  verified that $u_1\circ S_2(x_1,x_2,y_1,y_2,z_1,z_2,t_1,t_2)=x_1y_1\neq x_1$, namely $S_2^2\neq\id$, which means  that map \eqref{EN-map-2} is noninvolutive.

    Now, in view of \eqref{EN-map-1} and \eqref{EN-map-2} he have that
  $$
u_1v_1=x_1y_1,\quad u_1+v_1=x_1+y_1,\quad u_2v_2=x_2y_2,\quad w_1r_1=z_1t_1,
$$
and
 $$
    u_1v_1=x_1y_1,\quad u_2v_2=x_2y_2,\quad w_1r_1=z_1t_1,\quad w_1+r_1=z_1+t_1,
    $$
 The matrices with  rows $(\nabla I_i)^t$  and $(\nabla J_i)^t$, $i=1,2,3,4,$ have  rank  4, thus  invariants $I_i$ are functionally independent, and so are invariants $J_i$, $i=1,2,3,4.$

 Finally, $(x_1,y_1,z_1,t_1)\rightarrow(1,1,1,1)$, maps $S_i$, $i=1,2$, tend to the  trivial  4-simplex  extension  of  the electric  network tetrahedron map \eqref{EN},  therefore  $S_i$, $i=1,2$ can be regarded its 4-simplex exntensions.
\end{proof}

\section{Conclusions}\label{conclusions}
In this paper, we presented a scheme for constructing interesting 4-simplex extensions of Zamolodchikov tetrahedron maps; this scheme generalises the method presented in \cite{Sokor-2023-PhysD}. We demonstrated the advantage of this scheme using a Sergeev type tetrahedron map as an illustrative example. Moreover, we employed this scheme for the construction of novel Hirota type 4-simplex extensions, namely Bazhanov--Stroganov maps which are reduced to the famous Hirota map \eqref{Hirota} at a certain limit. Similarly, we contructed new 4-simplex maps which can be reduced to the electric network transform  \eqref{EN} at a certain limit.

Specifically, we constructed novel eight-dimensional Bazhanov--Stroganov maps, \eqref{Sergeev-map-1} and \eqref{Sergeev-map-2}, which can be regarded as 4-simplex extensions of Sergeev's map \eqref{Sergeev-b}. Furthermore, we constructed novel eight-dimensional 4-simplex maps, namely maps \eqref{Hirota-a}, \eqref{Hirota-b} and \eqref{Hirota-noninv}, which can be reduced to map \eqref{Hirota} at the limit $(x_2,y_2,z_2,t_2)\rightarrow (1,1,1,1)$. Finally,  we constructed the eight-dimensional  maps \eqref{EN-map-1}  and  \eqref{EN-map-2}, which are new solutions  to the functional 4-simplex equation,  and can be restricted to the famous  electric network 3-simplex map \eqref{EN}. It is worth noting that, even though maps \eqref{Hirota} and \eqref{EN} are involutive (thus they posseses trivial dynamics),  its 4-simplex extensions \eqref{Hirota-noninv} and \eqref{EN-map-2}, respectively, are noninvolutive.

The results of this paper can be extended in the following  ways.

\begin{itemize}
    \item All maps presented in this paper admit enough functionally independent first integrals which indicates their integrability. Their Liouvile interability is an open problem.

    \item We derived maps using particular $3\times3$ matrices. One could study all the possible $3\times3$ matrices that generate 4-simplex maps via  the local tetrahedron equation. Certain classification results on the solutions to the local tetrahedron equation which derive 4-simplex maps will appear in our future publication.

    \item Since the Hirota map \eqref{Hirota} is related to Desargues lattices \cite{Doliwa-Kashaev}, one could study the relation of the latter to the 4-simplex generalisations of the Hirota map \eqref{Hirota-a}, \eqref{Hirota-b} and \eqref{Hirota-noninv}.

    \item Study the relation of the derived Hirota type and electric network Bazhanov--Stroganov maps to 4D-lattice equations employing similar methods  which relate Yang--Baxter and tetrahedron maps to 2D and 3D lattice equations. For instance, the symmetries of the associated lattice equations \cite{Pap-Tongas-Veselov, Kassotakis-Tetrahedron}, the existence of integrals in separable variables \cite{Pavlos-Maciej-2}, lifts and squeeze downs \cite{Kouloukas-Dihn, Pap-Tongas, Sokor-Kouloukas}.

    \item Find which additional matrix conditions must be satisfied so that the solutions to the local tetrahedron equation define 4-simplex maps (see \cite{Sokor-2022} for similar results for tetrahedron maps).

    \item Derive new 4-simplex extensions of all the tetrahedron maps which were considered in \cite{Sokor-2023-PhysD}, including the Kadomtsev--Petviashvili tetrahedron map \cite{Dimakis-Hoissen} map and the NLS type tetrahedron maps \cite{Sokor-2020}, using the scheme presented in Section \ref{extension scheme}. 

    \item The method presented in Section \ref{extension scheme} can be extended by allowing matrix \eqref{K-matrix} depend on other auxiliary variables. For example, in \cite{Igonin-Sokor}, the tetrahedron map $(x,y,z)\stackrel{T}{\rightarrow}\left(y,x,\frac{yz}{x}\right)$ was generated by matrix {\small ${\rm L}(x)=\begin{pmatrix} x & 0 \\ 0 & \frac{1}{x}\end{pmatrix}$}. If ${\rm M}(x)={\rm L}^3_{12}(x)$, then one can consider matrix {\small ${\rm K}(x_1,x_2,x_3)=\begin{pmatrix} \frac{x_1}{x_2} & x_3 & 0 \\ 0 & \frac{1}{x_1} & 0\\ 0 & 0 & x_2\end{pmatrix}$}, such that ${\rm K}(x_1,x_2,x_3)\rightarrow {\rm M}(x_1)$, for $x_2\rightarrow 1, x_3\rightarrow 0$. Matrix ${\rm K}(x_1,x_2,x_3)$ generates via \eqref{local-tetrahedron} the following 4-simplex map
    $$
    (x_1,x_2,x_3,y_1,y_2,y_3,z_1,z_2,z_3,t_1,t_2,t_3)\rightarrow \left(y_1,y_2,\frac{x_1x_3z_2}{y_2z_1},x_1,x_2,\frac{y_2y_3z_1-x_1x_3z_2z_3}{x_1},\frac{y_1z_1}{x_1},z_2,z_3,\frac{t_1x_2}{y_2},t_2,t_3\right),
    $$
    which for $(x_2,x_3,y_2,y_3,z_2,z_3,t_2,t_3)\rightarrow (1,0,1,0,1,0,1,0)$ implies the tetrahedron map $(x,y,z)\stackrel{T}{\rightarrow}\left(y,x,\frac{yz}{x}\right)$.
    \end{itemize}

\section*{Acknowledgements} The work on Sections 2, 3 and 4 is funded by the Russian Science Foundation (Grant No. 21-71-30011), \url{https://rscf.ru/en/project/21-71-30011/}. The  work  on Sections 1, 5 and 6 was supported by the Ministry of Science and Higher Education of the Russian Federation (Agreement on provision of subsidy from the federal budget No. 075-02-2024-1442). I would like to thank S. Igonin for various useful discussions.

\end{document}